\documentclass[conference]{IEEEtran}
\IEEEoverridecommandlockouts
\usepackage{cite}
\usepackage{amsmath,amssymb,amsfonts}
\usepackage{graphicx}
\usepackage{textcomp}
\usepackage{xcolor}
\usepackage{subfigure,cite,graphicx,amsmath,amssymb,mathrsfs,epsf}
\usepackage{multirow}
\usepackage{tabularx}
\usepackage{stfloats}
\usepackage{algorithm}
\usepackage{algorithmic}
\usepackage{verbatim}
\usepackage{epsfig}
\usepackage{amsfonts}
\usepackage[normalem]{ulem}
\usepackage{amsthm}
\usepackage{float}
\usepackage{hyperref}
\usepackage{color}
\usepackage{mathtools,bm}
\usepackage[skip=0.5\baselineskip]{caption}
\usepackage[outdir=./]{epstopdf}
\IEEEoverridecommandlockouts

\normalsize
\newcommand{\BS}[1]{{\color{black} #1}}

\allowdisplaybreaks

\begin{document}
\title{Stochastic Coded Caching with Optimized Shared-Cache Sizes and Reduced Subpacketization}
%
\author{\IEEEauthorblockN{Adeel~Malik, Berksan~Serbetci, Petros~Elia}
\IEEEauthorblockA{Dept. of Communication Systems, EURECOM, Sophia Antipolis, 06410, France \\
\{malik, serbetci, elia\}@eurecom.fr}
\thanks{The work is supported by the European Research Council under the EU Horizon 2020 research and innovation program / ERC grant agreement no. 725929 (project DUALITY).}
}

\maketitle

\newtheorem{axiom}{Axiom}
\newtheorem{lemma}{Lemma}
\newtheorem{corollary}{Corollary}
\newtheorem{theorem}{Theorem}
\newtheorem{prop}{Proposition}
\newtheorem{observation}{Observation}
\newtheorem{definition}{Definition}
\newtheorem{remark}{Remark}
\newtheoremstyle{case}{}{}{}{}{}{:}{ }{}
\theoremstyle{case}
\newtheorem{case}{Case}

\begin{abstract}
This work studies the $K$-user broadcast channel with $\Lambda$ caches, when the association between users and caches is random, i.e., for the scenario where each user can appear within the coverage area of -- and subsequently is assisted by -- a specific cache based on a given probability distribution. Caches are subject to a cumulative memory constraint that is equal to $t$ times the size of the library. We provide a scheme that consists of three phases: the storage allocation phase, the content placement phase, and the delivery phase, and show that an optimized storage allocation across the caches together with a modified uncoded cache placement and delivery strategy alleviates the adverse effect of cache-load imbalance by significantly reducing the multiplicative performance deterioration due to randomness. In a nutshell, our work provides a scheme that manages to substantially mitigate the impact of cache-load imbalance in stochastic networks, as well as -- compared to the best known state-of-the-art -- the well-known subpacketization bottleneck by showing its applicability in deterministic settings for which it achieves the same delivery time -- which was proven to be close to optimal for bounded values of $t$ -- with an exponential reduction in the subpacketization.

\end{abstract}


\IEEEpeerreviewmaketitle

\setlength{\abovedisplayskip}{2pt}   
\setlength{\belowdisplayskip}{2pt}    

\section{Introduction}\label{Sec:1}
The ever-growing amounts of mobile data traffic have highlighted the need for innovative solutions that can provide service to an \BS{ever-increasing} number of users while using restricted network bandwidth resources. Within this framework, cache-enabled wireless networks have emerged as a viable option that can transfigure the storage capabilities of the network nodes into a fresh and powerful resource.

The seminal work in~\cite{man_it14} introduced the concept of coded caching, and revealed that an unbounded number of cache-aided users having different content requests can be served simultaneously by the aid of multicasting transmissions even with a bounded amount of network resources. \textcolor{black}{Key to this approach is a novel and carefully designed} cache placement algorithm. This delivery speedup is referred to as the \emph{coding gain} -- or equivalently as the \textcolor{black}{Degrees-of-Freedom (DoF) -- and it scales} with the total storage capacity of the \textcolor{black}{network.} \textcolor{black}{This same approach was shown} to be \textcolor{black}{information-theoretically optimal} in~\cite{wanOptimalityTransIT2020, yuTradeoff2TransIT2019} for the shared-link broadcast channel.

Many follow-up works have been studied since then, including the study of coded caching in D2D networks~\cite{ji_it16}, \textcolor{black}{in} subpacketization-constrained settings~\cite{yan_it17,lampiris_jsac18,shanmugam_it16}, \textcolor{black}{in settings with arbitrary} popularity distributions~\cite{man_it17,zhang_it18,serbetci_wiopt20}, and in other \textcolor{black}{settings as well\cite{malik_globecom20, brunero_arxiv21_comb, brunero_arxiv21_self, hachem_it17, parrinello_it20, malik_tcom21, wan_fogran_21}}.
\begin{figure}[t]
\centering
 \includegraphics[width=0.7\linewidth]{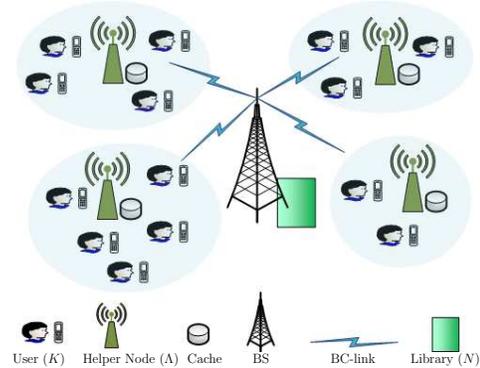}
\caption{An instance of a shared-cache network.}
\label{fig:SM}
\vspace{-.3cm}
\end{figure}

\subsection{Coded caching networks with shared caches}~\label{section:11}
Even though coded caching offers massive gains for the shared-link broadcast channel, its applicability to more realistic settings requires further exploration. In the so-called shared-cache setting, different users are within the coverage area of different cache-aided base stations (or caches), and are forced to benefit from the same cache content stored in the cache that they are associated with. \BS{This shared-cache setting is of great importance, and it arguably represents a much more realistic scenario as having all users with their own caches dedicated to a certain content library is a fanciful assumption in practical wireless communication settings}\textcolor{black}{\cite{parrinello_it20, malik_tcom21}}. \textcolor{black}{Such realistic settings may include} cache-enabled heterogeneous networks (HetNets), where a central transmitter (a base station, or a macrocell) needs to deliver content to a set of interfering users with the assistance of cache-enabled helper nodes (small base stations, or femtocells), where the cache content is available to all users (cost-free) within the coverage area of that node. An instance of a shared-cache network is shown in Figure~\ref{fig:SM}. 

An early work on this scenario can be found in~\cite{hachem_it17}, where each helper node is \textcolor{black}{serving an} equal number of users. The work in~\cite{parrinello_it20} removed this assumption, and characterized the optimal (under the uncoded cache placement) worst-case delivery time for the case when an arbitrary number of users is assisted by each cache, and when the cache placement is agnostic to the user-to-cache association. \BS{Subsequently,} the work in~\cite{malik_tcom21} extended the deterministic user-to-cache association setting in~\cite{parrinello_it20} to the stochastic network setting, \textcolor{black}{and for the setting where} the cache populations \textcolor{black}{follow} a given probability distribution. \textcolor{black}{This work revealed the surprising fact that} the cache-load imbalance could lead to the significant deterioration in coding gains. \textcolor{black}{This same work also showed that} a more balanced user-to-cache association performs better than the unbalanced one, leading to the fact that the multiplicative performance deterioration caused by the stochastic nature of the problem can be drastically reduced by carefully tuning the cache-load balance.

The work in~\cite{wan_fogran_21} proposed a novel coded placement strategy by assuming that user-to-cache association during the placement phase is known, and showed that exploiting this extra knowledge yields additional coding gains. In a similar \textcolor{black}{context}, the work in~\cite{parrinello_itw19} optimized the cache sizes as a function of the number of users served by each cache, and then proposed a novel coded caching scheme that outperforms the optimal scheme in~\cite{parrinello_it20}. 


At this point, we need to highlight that there is a \textcolor{black}{direct} relation between the aforementioned shared-cache setting and the so-called subpacketization bottleneck. In order to achieve \textcolor{black}{the originally promised \emph{theoretical coding gains}~\cite{man_it14}}, each file in the content library must be partitioned into $S$ unit-size subpackets, where $S$ scales exponentially with number of cache states\footnote{\textcolor{black}{When adopting the cache placement strategy in~\cite{man_it14} for $\Lambda$ cache-enabled users, we refer the content to be placed in a single cache as a cache state.}}. \BS{Subsequently,} a subset of these subpackets is cached at different nodes depending on the cache-enabled device's cache state. The number of distinct cache states, denoted as $\Lambda$, is then subject to some physical limitations, i.e., a file cannot be partitioned into more subpackets than a certain threshold, hence forcing $S$ to be less than some certain number, and inevitably forcing $\Lambda$ to be less than a certain value. In a nutshell, $\Lambda$ must be generally less than the total number of users since it is known that traditional coded caching techniques require file sizes that scale exponentially with $\Lambda$ (cf.~\cite{lampiris_jsac18,shanmugam_it16}.). In broad terms, reducing the number of cache states leads to a reduction in the coding gain, \BS{hinting out that the \emph{subpacketization bottleneck} is in fact a major factor on the performance.}

\BS{In this work, we aim to exploit cache-size differences and optimize the individual cache sizes to mitigate the cache-load imbalance bottleneck in stochastic shared-cache networks. To do so, we propose a coded caching scheme that optimizes the individual cache sizes based on each cache's load \textcolor{black}{statistics,} subject to a given cumulative cache capacity. We characterize the performance of our scheme and numerically verify its effectiveness in \textcolor{black}{substantially ameliorating} the impact of cache-load imbalance on the coding gain. We also show that for a \BS{deterministic} user-to-cache association setting, our scheme achieves the same state-of-the-art \textcolor{black}{(SoA)} delivery time as of~\cite{parrinello_itw19} with a significant (exponential) reduction in subpacketization, thus making our scheme more suitable than~\cite{parrinello_itw19} to apply in HetNets in the finite file size regime.}

\subsection{Notations}~\label{section:16}
Throughout this paper, we use the notation $[z]\triangleq$ $\left[1, 2, \dots, z \right]$, and we use $\mathbf{Q} \backslash \mathbf{P}$ to denote the set difference of 
$\mathbf{P}$ and $\mathbf{Q}$, which is the set of elements in $\mathbf{Q}$ but not in $\mathbf{P}$. We use $\textcolor{black}{\mathcal{X}^{[n]}_{k}\triangleq} \left\{\delta : \delta \subseteq [n], \left|\delta \right|=k \right\}$ and \textcolor{black}{we use $\delta(i)$ to denote} the $i$th element of $\delta$. 


\section{Network Setting}\label{Sec:model}
We consider a heterogeneous \textcolor{black}{cache-aided network} setting which consists of a base station (BS) having access to a library of $N$ unit-sized files $\mathcal{F}=\left\{F^1, F^2, \dots, F^N \right\}$, \textcolor{black}{as well as consists of} $\Lambda$ cache-enabled helper nodes (i.e., caches), and $K$ \textcolor{black}{receiving} users. 
The BS delivers content via an error-free broadcast link of bounded capacity per unit of time to $K$ users, with the assistance of helper nodes. We assume that users within the coverage area of a cache $\lambda \in [\Lambda]$ have direct access to the content stored at that cache. We consider the scenario of non-uniform cache population intensities \textcolor{black}{where the number of users served by each cache may not be identical.} For any cache $\lambda \in [\Lambda]$, let $p_{\lambda}$ be the probability that a user appears in the coverage area \textcolor{black}{of this} $\lambda$th cache-enable helper node, \textcolor{black}{and let} $\mathbf{p} =\left[p_1, p_2, \dots, p_{\Lambda}\right]$ \textcolor{black}{denote the cache population intensities vector}, where $\sum_{\lambda \in [\Lambda]}p_{\lambda}=1$. Without loss of generality, we assume that $p_{1}\geq p_{2} \geq \dots \geq p_{\Lambda}$. At any given instance, we denote $\mathbf{V} =\left[v_1, v_2, \dots,v_{\Lambda} \right]$ to be the cache population vector, where $v_{\lambda}$ is the number of users \textcolor{black}{having} access to the content of cache $\lambda \in[\Lambda]$, \textcolor{black}{and we let} $\bar{\mathbf{V}} = K\textbf{p}=\left[\bar{v}_1, \bar{v}_2, \dots,\bar{v}_{\Lambda} \right]$ \textcolor{black}{be} the expected cache population vector. 

The size of each cache $\lambda \in \left[\Lambda\right]$ is \textcolor{black}{the} design parameter $M_{\lambda} \in (0,N]$ \textcolor{black}{(measured in units of file),} adhering to a cumulative sum cache-size constraint $\sum_{\lambda = 1}^{\Lambda} M_{\lambda} = M_{\Sigma}$. For any cache $\lambda \in [\Lambda]$, we denote by $\gamma_{\lambda} \triangleq \frac{M_{\lambda}}{N}$ the normalized cache capacity, and subsequently \textcolor{black}{we have the} normalized cache capacity vector $\bm{\gamma}= [\gamma_{1}, \gamma_{2}, \dots,  \gamma_{\Lambda}]$. Consequently, the normalized cumulative cache-size constraint takes the form
\begin{equation}
\sum_{\lambda=1}^{\Lambda}\gamma_{\lambda}=t\triangleq\frac{M_{\Sigma}}{N}.\label{eq:budget}
\end{equation}
The communication process consists of three phases; the \emph{storage allocation phase}, the \emph{content placement phase} and the \emph{delivery phase}.
The storage allocation phase involves allocating the cumulative cache capacity $M_{\Sigma}$ (or the normalized cumulative cache capacity $t$) to the caches subject to~\eqref{eq:budget}, 
\textcolor{black}{a process that results in the aforementioned normalized} cache capacity vector $\bm{\gamma}$. The content placement phase involves the placement of a portion of library-content, $\mathcal{Z}_{\lambda}$, in each cache $\lambda \in [\Lambda]$ \textcolor{black}{--- respecting} its allocated cache capacity $\gamma_{\lambda}$ \textcolor{black}{--- according to} a certain placement strategy $\mathcal{Z}=[\mathcal{Z}_1, \mathcal{Z}_2, \dots, \mathcal{Z}_{\Lambda}]$. We assume \textcolor{black}{that the} first two phases are aware \textcolor{black}{of the} cache population intensities $\textbf{p}$. However, 
\textcolor{black}{these two phases} are oblivious to the actual requests generated during the delivery phase. The delivery phase begins after each user $k\in [K]$ appears within the coverage area of one of the caches, and requests a content file $F^{d_k}\in \mathcal{F}$, where $d_k$ is the index of the file requested by user $k\in[K]$. Then for any \emph{request vector} $\mathbf{d}=\left[d_1, \dots, d_{v_1}, d_{v_1+1},\dots, d_K \right]$, using the knowledge of the content stored at each cache $\mathcal{Z}_{\lambda} \in \mathcal{Z}$, and the user-to-cache association, the BS delivers the content to the user. 
\subsection{Problem Definition} 
For a given normalized \textcolor{black}{cache-size budget} $t$, and cache population intensities vector $\textbf{p}$, our goal is to design a content placement strategy $\mathcal{Z}$, and a delivery scheme for the system where a BS is serving $K$ users with the help of $\Lambda$ caches. \textcolor{black}{Then our goal is to} evaluate its performance in terms of \textcolor{black}{the} time needed to complete the delivery of any request vector $\mathbf{d}$. 
\textcolor{black}{Given the random nature of our problem where at any given} instance of the problem we may experience a different cache population vector $\mathbf{V}$, our measure of interest is the average delay
\begin{align}
 \overline{T}(t)=E_{\mathbf{V}}[T(\mathbf{V})]= \sum_{\mathbf{V} \in \mathcal{V}} P(\mathbf{V})T(\mathbf{V}),
\end{align}
where $T(\mathbf{V})$ is the worst-case time needed to complete the delivery of any request vector $\mathbf{d}$ corresponding to a specific cache population vector $\mathbf{V}$\textcolor{black}{, where} $\mathcal{V}$ is the set of all possible cache population vectors, \textcolor{black}{and where} $P(\mathbf{V})$ is the probability of observing the cache population vector $\mathbf{V}$.
\section{Main Result}
In this section, we first present our main results on the performance of the stochastic network setting described in Section \ref{Sec:model}. \BS{After doing so, we will also} discuss the applicability as well as the efficacy of our proposed scheme for a deterministic shared cache setting studied in \cite{parrinello_itw19} \BS{by showing that it provides an exponential reduction in the subpacketization}.
Our first result is the characterization of the achievable delivery time $T(\mathbf{V})$ for any cache population vector $\mathbf{V}$. Crucial to this characterization is the greatest common divisor (\textcolor{black}{GCD}) of vector $\bar{\mathbf{V}}$, which we denote by $\alpha$, and the partition of $\mathbf{V}$ into a set $\mathcal{B}_V=[\mathbf{V}^1, \mathbf{V}^2, \dots, \mathbf{V}^{\beta_\mathbf{V}}]$ of $\beta_\mathbf{V}=\max_{i\in [\Lambda]} \frac{\alpha v_{i}}{\bar{v_i}}$ vectors according to the partition algorithm presented in Algorithm~\ref{alg:partV}. \textcolor{black}{Each} resulting partition vector $\mathbf{V}^j=[v^j_1, v^j_2, \dots, v^j_{\Lambda}]$ will then satisfy $v^j_\lambda \leq \frac{\bar{ v_{\lambda}}}{\alpha}$, where $\sum_{j \in [\beta_\mathbf{V}] } v^j_\lambda=v_\lambda$, $\forall \lambda \in [\Lambda]$. Under the assumption of a random user-to-cache association with cache population intensities vector $\mathbf{p}$ such that the expected cache population $\bar{v}_{\lambda} \in \bar{\mathbf{V}}$ is a non-negative integer for each cache $\lambda= [\Lambda]$, the achievable delay is given in the following theorem.
\begin{theorem}\label{th:tv}
In the $K$-user, $\Lambda$-cache setting with a normalized cache budget $t$, and a random user-to-cache association with cache population intensities vector $\mathbf{p}$, the delivery time for any cache population vector $\mathbf{V}$
\begin{align} \label{eq:tv}
T(\mathbf{V})= \!\!\! \sum_{j \in [\beta_\mathbf{V}] }\!\!\!\frac{\sum\limits_{\tau \in \mathcal{X}_{t+1}^{[\Lambda]}} \prod\limits_{i=1}^{t+1} \frac{\bar{v}_{\tau(i)}}{\alpha} - \!\!\!\!\! \sum\limits_{\tau \in \mathcal{X}_{t+1}^{\mathcal{A}_j}  } \prod\limits_{i=1}^{t+1}\left(\frac{\bar{v}_{\tau(i)}}{\alpha}-v^j_{\tau(i)}\right)} {\sum\limits_{\tau \in \mathcal{X}_{t}^{[\Lambda]}} \prod\limits_{i=1}^{t} \frac{\bar{v}_{\tau(i)}}{\alpha}}
\end{align}
is achievable if the expected cache population vector $\bar{\mathbf{V}}$ is a non-negative integer vector, where $\mathcal{A}_j \subseteq [\Lambda]$ is a subset of the set of caches, such that for each cache $\lambda \in \mathcal{A}_j$, $\frac{\bar{v}_{\lambda}}{\alpha}>v^j_{\lambda}$.
\end{theorem} 
\begin{proof} The proof is deferred to~Section~\ref{section:scheme}. 
\end{proof}
With Theorem \ref{th:tv} in hand, we present our next result, which is the average delay $\overline{T}(t)$ corresponding to our stochastic network setting.

\begin{theorem}\label{th:tavg}
In the $K$-user, $\Lambda$-cache setting with a normalized cache budget $t$, and a random user-to-cache association with cache population intensities vector $\mathbf{p}$, the average delay of
\begin{align} \label{eq:tavg}
\overline{T}(t)= \sum_{\mathbf{V}\in \mathcal{V}} \frac{T(\mathbf{V}) K!}{\prod_{\lambda \in [\Lambda]}v_{\lambda}} \prod_{\lambda \in [\Lambda]} p_{\lambda}^{v_{\lambda}} 
\end{align}
is achievable if the expected cache population vector $\bar{\mathbf{V}}$ is a non-negative integer vector.
\end{theorem} 
\begin{proof} From the fact \textcolor{black}{that the} probability distribution $P(\mathbf{V})$ follows the well-known multinomial distribution with parameter $\mathbf{p}$, we have 
\begin{align}\label{eq:pv} 
P(\mathbf{V}) = \frac{K!}{\prod_{\lambda \in [\Lambda]}v_{\lambda}} \prod_{\lambda \in [\Lambda]} p_{\lambda}^{v_{\lambda}}.
\end{align}
Combining \eqref{eq:tv} with \eqref{eq:pv} allows us to obtain \eqref{eq:tavg}, which concludes the proof. 
\end{proof}
Next, we see the applicability of our scheme in a similar setting, but with a fixed user-to-cache association, which was initially studied in \cite{parrinello_itw19}. 

\begin{corollary}\label{co:tmean}
In the $K$-user, $\Lambda$-cache setting with a normalized cache budget $t$, and a fixed user-to-cache association with cache population vector $\bar{\mathbf{V}}$, the proposed scheme in Section \ref{section:scheme} achieves the delivery time of
\begin{align} \label{eq:tmean}
T(\bar{\mathbf{V}})=  \alpha \frac{\sum\limits_{\tau \in \mathcal{X}_{t+1}^{[\Lambda]}} \prod\limits_{i=1}^{t+1} \frac{\bar{v}_{\tau(i)}}{\alpha} } {\sum\limits_{\tau \in \mathcal{X}_{t}^{[\Lambda]}} \prod\limits_{i=1}^{t} \frac{\bar{v}_{\tau(i)}}{\alpha}},
\end{align}
which is same as of \cite[equation (11)]{parrinello_itw19} and it requires the \textcolor{black}{subpacketization} rate of
\begin{align} \label{eq:S}
S=  \sum_{\tau \in \mathcal{X}_{t}^{[\Lambda]}} \prod_{j=1}^{t} \frac{\bar{v}_{\tau(j)}}{\alpha},
\end{align}
which is $\alpha^t$ times less than the \textcolor{black}{subpacketization} rate of~\cite[equation (7)]{parrinello_itw19}.
\end{corollary} 
\begin{proof} The proof is straightforward from~\eqref{eq:S_disc} and~\eqref{eq:Tbase}.
\end{proof}

\textcolor{black}{Corollary \ref{co:tmean} \textcolor{black}{reveals the substantial benefits of our scheme compared to the \textcolor{black}{SoA}}~\cite{parrinello_itw19} as it achieves the same delivery time with a significantly reduced -- with a factor of $\alpha^t$ --  subpacketization. This exponential reduction in subpacketization is a crucial contribution as the subpacketization is a major bottleneck in the applicability of coded caching schemes\cite{lampiris_jsac18, malik_tcom21} especially in the finite file size regimes. To illustrate this gain, in Figure~\ref{fig:Subpacket}, we compare the required subpacketization of our scheme with the scheme in \cite{parrinello_itw19} for $\bar{\mathbf{V}}=[8,6,6,4,2,2]$.} We can see that even for a modest network consisting of an extremely small number of users, our scheme requires significantly less subpacketization. 
\begin{figure}[t]
\centering
 \includegraphics[width=0.99\linewidth]{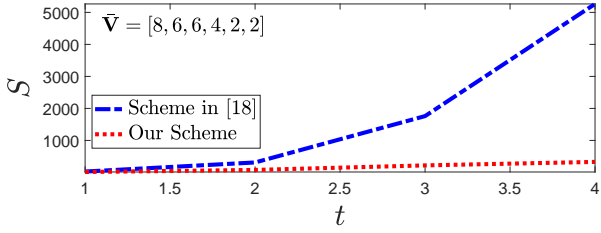}
\caption{Comparison of the required subpacketization.}
\label{fig:Subpacket}
\end{figure}
\begin{remark}~\label{R:0}
\BS{An interesting observation of our scheme is that $\alpha$ can be treated as a trade-off parameter between the subpacketization and the delivery time. For example, when $\bar{\mathbf{V}}=[20,15,15,5,5,4]$, we have $\alpha=1$. However, if we associate one virtual user to the $\Lambda$-th cache, $\alpha$ increases from $1$ to $5$, thus, reducing the subpacketization by a factor of $5^t$ with a modest increase in the delivery time. We leave the study of this trade-off for future work.}
\end{remark}

\section{Placement and delivery}\label{section:scheme}
In this section, we present the proof of Theorem \ref{th:tv}. We describe the content placement and delivery strategies that can achieve the average delay of \eqref{eq:tv}.
\subsection{Content Placement}~\label{section:21}
The content placement scheme is based on the idea of assigning more storage capacity to the caches with high population intensities. \textcolor{black}{We denote $\hat{\mathbf{V}}=[\hat{v}_1, \hat{v}_2, \dots, \hat{v}_{\Lambda}]$ as the \emph{base cache population vector}, which is given as $\hat{\mathbf{V}}\triangleq \frac{\bar{\mathbf{V}}}{\alpha}$, where $\alpha$ is the GCD of the elements in $\bar{\mathbf{V}}$. For a base cache population vector $\hat{\mathbf{V}}$, we assume that there are $\hat{\Lambda}\triangleq \sum_{\lambda=1}^{\Lambda} \hat{v}_{\lambda}$ virtual caches such that each cache $\lambda \in [\Lambda]$ consists of $\hat{v}_{\lambda}$ virtual caches. We then use $\mathcal{C}= [1, 2, \dots,  \hat{\Lambda} ]$ to denote the set of virtual caches, and $\mathcal{C}_{\lambda}= [\hat{v}_{\lambda-1}+1, \dots, \hat{v}_{\lambda-1}+\hat{v}_{\lambda}]$ (assuming $\hat{v}_{0}\triangleq 0$) to denote the set of virtual caches that belongs to cache $\lambda \in[\Lambda]$.} Let $\mathcal{Q}_{t}^{\mathcal{C}}\subseteq \mathcal{X}_{t}^{\mathcal{C}}$ be the set of all possible \textcolor{black}{$t$-tuples of} $\mathcal{C}$ such that for each tuple $\tau \in \mathcal{Q}_{t}^{\mathcal{C}}$, no two virtual caches $i$, $j \in \tau$ belong to the same cache $\lambda \in [\Lambda]$. Next, each content file $F^i \in \mathcal{F}$ is divided into $\left|\mathcal{Q}_{t}^{\mathcal{C}}\right|$ subpackets, and labeled as $F^i= \left\{ F_{\tau}^i \right\}_{\tau \in \mathcal{Q}_{t}^{\mathcal{C}}}$.  Then, the set of contents to be cached at each cache $\lambda \in [\Lambda]$ is given by 
\begin{align}
  \mathcal{Z}_{\lambda}= \left\{\mathcal{Z}_{\hat{\lambda}}: \hat{\lambda} \in \mathcal{C}_{\lambda}\right\}, 
\end{align}
 where 
 \begin{align}
    \mathcal{Z}_{\hat{\lambda}} = \left\{F_{\tau}^i : F_{\tau}^i \in F^i,    \tau \in \mathcal{Q}_{t}^{\mathcal{C}}, \hat{\lambda} \in \tau, F^i \in \mathcal{F} \right\}.  
 \end{align}
 From the fact that for each cache $\lambda$ we have $\hat{v}_{\lambda}$ virtual caches, \textcolor{black}{we can conclude that for any} $t$-tuple of caches $\tau \in \mathcal{X}_{t}^{[\Lambda]}$ there \textcolor{black}{must be} $\prod_{i=1}^{t}\hat{v}_{\tau(i)}$ $t$-tuples of virtual caches that belong to the set $\mathcal{Q}_{t}^{\mathcal{C}}$. Consequently, the number of $t$-tuples of $\mathcal{C}$ in the set  $\mathcal{Q}_{t}^{\mathcal{C}}$ is given as
\begin{align}\label{eq:S_disc}
 \left|\mathcal{Q}_{t}^{\mathcal{C}}\right|= \sum_{\tau \in \mathcal{X}_{t}^{[\Lambda]}} \prod_{j=1}^{t} \hat{v}_{\tau(j)},   
\end{align}
 and the normalized cache capacity required at any cache $\lambda \in [\Lambda]$ is given as 
\begin{align}\label{eq:gamma}
 \gamma_{\lambda} =\frac{\left|\mathcal{Z}_{\lambda}\right|}{N\left|\mathcal{Q}_{t}^{\mathcal{C}}\right|}= \frac{  \sum\limits_{\tau \in \mathcal{X}_{t}^{[\Lambda]}: \tau \ni \lambda  } \prod\limits_{j=1}^{t} \hat{v}_{\tau(j)} } {  \sum\limits_{\tau \in \mathcal{X}_{t}^{[\Lambda]}} \prod\limits_{j=1}^{t} \hat{v}_{\tau(j)}}.
\end{align}
Thus, \eqref{eq:gamma} yields our proposed storage allocation strategy.
We can see that $\sum_{\lambda\in [\Lambda]} \gamma_{\lambda} = t$, as for each \textcolor{black}{$\tau \in \mathcal{Q}_{t}^{C}$}, the corresponding subpackets $\left\{F_{\tau}^i : F^i \in \mathcal{F} \right\}$ are placed in $t$ caches. Thus, the content placement strategy satisfies the total caching budget constraint~\eqref{eq:budget}.  
\begin{algorithm}[h]
  \caption{Cache Population Vector Partition}
  \label{alg:partV}
  \begin{algorithmic}
 \STATE\textbf{Input:} $\mathbf{V}$ and $\bar{\mathbf{V}}$ \\
\STATE \textbf{Output:} $\mathcal{B}_V$ \\
\STATE \textbf{Initialization:} \!$\alpha\leftarrow GCD(\bar{\mathbf{V}})$, $\beta_{\mathbf{V}}\! \leftarrow\! \max\limits_{i\in [\Lambda]} \frac{\alpha v_{i}}{\bar{v_i}}$, $\mathcal{B}_V \leftarrow \phi$ \\
\STATE \textbf{for} $j$ \textbf{from} 1 \textbf{to} $\beta_{\mathbf{V}}$ \textbf{do} \\
\STATE \ \ \textbf{for} $i$ \textbf{from} 1 \textbf{to} $\Lambda$ \textbf{do} \\
\STATE \ \ \ \ \textbf{if} $v_i > \frac{\bar{v}_i}{\alpha}$\\
\STATE \ \ \ \ \ \ $v^j_i  \leftarrow \frac{\bar{v}_i}{\alpha} $, \ \ \ $v_i  \leftarrow v_i- \frac{\bar{v}_i}{\alpha} $ \\
\STATE \ \ \ \ \textbf{else}\\
\STATE \ \ \ \ \ \ $v^j_i  \leftarrow v_i$, \ \ \  $v_i  \leftarrow 0$
\STATE \ \ \ \ \textbf{end if} \\
\STATE \ \ \textbf{end for} \\
\STATE \ \ $\mathcal{B}_V  \leftarrow [\mathcal{B}_V, [v^j_1, v^j_2, \dots, v^j_{\Lambda}]]$
\STATE \textbf{end for} \\
  \end{algorithmic}
\end{algorithm}

\subsection{Content Delivery}
\textcolor{black}{We will consider the worst-case delivery scenario where each user} requests a different file. Once the BS is notified of the cache population vector $\mathbf{V}$ and the corresponding request vector $\mathbf{d}$, it commences delivery. We propose a delivery scheme that is completed in $\beta_\mathbf{V}=\max_{i\in [\Lambda]} \frac{v_{i}}{\hat{v}_i}$ rounds, where the content is delivered to at most $\hat{v}_{\lambda}$ users from each cache $\lambda \in [\Lambda]$ in each round. We divide the cache population vector $\mathbf{V}$ into a set of $\beta_\mathbf{V}$ vectors $\mathcal{B}_V= [\mathbf{V}^1, \mathbf{V}^2, \dots, \mathbf{V}^{\beta_V}]$ based on the procedure described in Algorithm \ref{alg:partV}, such that for all $j\in[\beta_V]$, $\mathbf{V}^j=[v^j_1, v^i_2, \dots, v^j_{\Lambda}]$ \BS{satisfies} $v^j_\lambda \leq \hat{v}_{\lambda}$ and $\sum_{j \in [\beta_\mathbf{V}] } v^j_\lambda=v_\lambda$ for all $\lambda \in [\Lambda]$. \textcolor{black}{In the following, we describe our delivery strategy for the two only possible cases after applying Algorithm~\ref{alg:partV}.}



\textbf{Case 1: $\mathbf{V}^j= \hat{\mathbf{V}}$}: \textcolor{black}{In this case, the BS will serve $\hat{\Lambda}$ users based on the base cache population vector $\hat{\mathbf{V}}$.} Let $\mathbf{U}=[u_1, u_2, \dots, u_{\hat{\Lambda}}]$ denote the set of indices of users that corresponds to $\hat{\mathbf{V}}$, and $\mathbf{U}_{\lambda}=\left\{u_i\right\}_{i=\hat{v}_{\lambda-1}+1}^{\hat{v}_{\lambda-1}+\hat{v}_{\lambda}}$ be the set of users associated \textcolor{black}{to} cache $\lambda \in [\Lambda]$ (assuming $\hat{v}_{0}=0$). The corresponding request vector is $\mathbf{d^{\hat{\mathbf{V}}}}=\left[d_{u_1}, d_{u_2}, \dots, d_{u_{\hat{\Lambda}}} \right]$. Let $\mathcal{Q}_{t+1}^{\mathcal{C}}\subseteq \mathcal{X}_{t+1}^{\mathcal{C}}$ be the set of \textcolor{black}{all possible $(t+1)$-tuples of $\mathcal{C}$} such that for each tuple $\tau \in \mathcal{Q}_{t+1}^{\mathcal{C}}$, no two virtual caches $i$, $j \in \tau$ belong to the same physical cache. Then, for each $(t+1)$-tuple $\tau \in \mathcal{Q}_{t+1}^{\mathcal{C}}$, the BS transmits the \textcolor{black}{following XOR:}
\begin{align}
    \mathcal{Y}_{\tau}= \oplus_{\hat{\lambda} \in \tau} F^{d_{u_{\hat{\lambda}}}}_{\tau \backslash \hat{\lambda}}. 
\end{align}
The structure of $\mathcal{Y}_{\tau}$ allows to serve $t+1$ users simultaneously as each user can easily decode its required subpacket using the content $\mathcal{Z}_{\lambda}$ of its associated cache $\lambda\in [\Lambda]$. Let $\mathcal{Y}= \left\{\mathcal{Y}_{\tau}:\tau \in \mathcal{Q}_{t+1}^{\mathcal{C}} \right\}$ denote the set of all transmissions. In order to completely serve the request vector $\mathbf{d^{\hat{\mathbf{V}}}}$ corresponding to cache population vector $\hat{\mathbf{V}}$, the BS transmits $\left|\mathcal{Y}\right|=\left|\mathcal{Q}_{t+1}^{\mathcal{C}}\right|= \sum\limits_{\tau \in \mathcal{X}_{t+1}^{[\Lambda]}} \prod\limits_{j=1}^{t+1} \hat{v}_{\tau(j)}$ XORs in the set $\mathcal{Y}$. Thus the corresponding transmission delay is given as 
\begin{align}\label{eq:Tbase}
T(\hat{\mathbf{V}}) =\frac{\sum\limits_{\tau \in \mathcal{X}_{t+1}^{[\Lambda]}} \prod\limits_{i=1}^{t+1} \hat{v}_{\tau(i)} }{\sum\limits_{\tau \in \mathcal{X}_{t}^{[\Lambda]}} \prod\limits_{i=1}^{t} \hat{v}_{\tau(i)}}.
\end{align}
\textbf{Case 2: $\mathbf{V}^j \ni v^j_{\lambda}< \hat{v}_{\lambda} \text{ for some } \lambda \in [\Lambda]$}: 
\textcolor{black}{The delivery scheme for this case is exactly the same as for the case when $\mathbf{V}^j = \hat{\mathbf{V}}$. However, the number of users to be served in this case is less than $\hat{\Lambda}$, i.e., $\vert\mathbf{U}\vert < \hat{\Lambda}$.} Hence, there may exists some subpackets in $\mathcal{Y}$ that \textcolor{black}{does not} serve any user, \textcolor{black}{and the BS only transmits} the subpacket $\mathcal{Y}_{\tau}\in\mathcal{Y}$ if it serves at least one user. \textcolor{black}{Let $\mathcal{A}_j\!\subseteq\![\Lambda]$ be the subset of caches such that for each cache $\lambda\!\in\!\mathcal{A}_j$, $v^j_{\lambda}\!<\!\hat{v}_{\lambda}$ holds. Then, for any cache population vector $\mathbf{V}^j$, the total number of subpackets $\mathcal{Y}_{\tau}\in\mathcal{Y}$ that will not serve any user is given by \!\!\!$\sum\limits_{\tau \in \mathcal{X}_{t+1}^{\mathcal{A}_j}  } \prod\limits_{i=1}^{t+1} (\hat{v}_{\tau(i)}-v^j_{\tau(i)})$.
}
\textcolor{black}{Consequently,} for any cache population vector $\mathbf{V}^j \ni v^j_{\lambda}\leq \hat{v}_{\lambda}$ $\forall \ \lambda \in [\Lambda]$, the  transmission delay $T(\mathbf{V}^j)$ is given as
\begin{align}
T(\mathbf{V}^j) =       \frac{\sum\limits_{\tau \in \mathcal{X}_{t+1}^{[\Lambda]}} \prod\limits_{i=1}^{t+1} \hat{v}_{\tau(i)} -  \sum\limits_{\tau \in \mathcal{X}_{t+1}^{\mathcal{A}_j}  } \prod\limits_{i=1}^{t+1} (\hat{v}_{\tau(i)}-v^j_{\tau(i)})} {\sum\limits_{\tau \in \mathcal{X}_{t}^{[\Lambda]}} \prod\limits_{i=1}^{t} \hat{v}_{\tau(i)}}.
\end{align}
\textcolor{black}{Hence,} the transmission delay $T(\mathbf{V})$ corresponding to the cache population vector $\mathbf{V}$ is equal to \eqref{eq:tv}.

\textbf{Example:} Let us take \textcolor{black}{the} example of $K=N=10$, $\Lambda=4$, $t=2$, and $\mathbf{p}=(0.4, 0.2, 0.2, 0.2)$. Then the expected cache population vector is $\bar{\mathbf{V}}=[4, 2, 2, 2]$, and consequently $\hat{\mathbf{V}}=[2, 1, 1, 1]$ ($\alpha=2$ for $\bar{\mathbf{V}}$). The set of virtual caches is $\mathcal{C}=[1, 2, 3, 4, 5]$, where the virtual caches $\mathcal{C}_1=[1,2]$, $\mathcal{C}_2=[3]$, $\mathcal{C}_3=[4]$, and $\mathcal{C}_4=[5]$ belong to the caches 1, 2, 3, and 4 respectively. Then we have nine 2-tuples in the set $\mathcal{Q}_{t}^{\mathcal{C}}=$ $[(1,3), (1,4), (1,5), (2,3),$ $ (2,4), (2,5), (3,4), (3,5), (4,5)]$ and each file is divided into nine subpackets. The content placement at each cache is given as \\
$\mathcal{Z}_{1}\!=\!\left\{\!F^i_{\tau}\!:\!\tau\!\in\![(1,3),\!(1,4),\!(1,5),\!(2,3), (2,4), (2,5)],  i\!\in\![N]\!\right\}$, \\ 
$\mathcal{Z}_{2}\!=\!\left\{\!F^i_{\tau}\!:\!\tau\!\in\![(1,3),\!(2,3),\!(3,4),\!(3,5)], i\!\in\![N]\!\right\}$,\\
$\mathcal{Z}_{3}\!=\!\left\{\!F^i_{\tau}\!:\!\tau\!\in\![(1,4),\!(2,4),\!(3,4),\!(4,5)], i\!\in\![N]\!\right\}$,\\
$\mathcal{Z}_{4}\!=\!\left\{\!F^i_{\tau}\!:\!\tau\!\in\![(1,5),\!(2,5),\!(3,5),\!(4,5)], i\!\in\![N]\!\right\}$.\\
This placement leads to the normalized cache capacity allocation of $\gamma= [\frac{6}{9},\frac{4}{9}, \frac{4}{9}, \frac{4}{9} ]$. 
Now, let us move to the content delivery phase. Let us assume the case of $\mathbf{V}=[6,2,1,1]$, where users 1 to 6 have access to cache 1 and \textcolor{black}{request} the content $F^1$, $F^2$, $F^3$, $F^4$, $F^5$, and $F^6$ respectively, users 7 and 8 have access to cache 2 and request $F^7$ and $F^8$ respectively, user 9 has access to cache 3 and requests $F^9$, and user 10 has access to cache 4 and requests $F^{10}$.  The BS partitions the cache population vector into the set of $\beta_\mathbf{V}=3$ vectors, and transmits the content in $3$ rounds. The partition of $\mathbf{V}$ based on Algorithm \ref{alg:partV} leads to $\mathcal{B}_V= [(2,1,1,1), (2,1,0,0), (2,0,0,0)]$, and the corresponding demand vectors are $\mathbf{d}^{\mathbf{V}^1}=[1,2,7,9,10]$, $\mathbf{d}^{\mathbf{V}^2}=[3,4,8]$, and $\mathbf{d}^{\mathbf{V}^3}=[5,6]$. 
In the first round of delivery, the BS serves user 1 and 2 from cache 1, user 7 from cache 2, user 9 from cache 3, and user 10 from cache 4. For this round, we have $\mathbf{V}^1= \hat{\mathbf{V}}$, thus, for the demand vector $\mathbf{d}^{\mathbf{V}^1}=[1,2,7,9,10]$, the BS transmits the following seven subpackets based on set $\mathcal{Q}_{t+1}^{\mathcal{C}}=[(1,\!3,\!4), (1,\!3,\!5), (1,\!4,\!5), (2,\!3,\!4), (2,\!3,\!5), (2,\!4,\!5), (3,\!4,\!5)]$,
\begin{align*}
\mathcal{Y}_{1,3,4}&\!=\!F^1_{3,4} \oplus F^7_{1,4} \oplus F^9_{1,3},\ \  \mathcal{Y}_{1,3,5}\!=\!F^1_{3,5} \oplus F^7_{1,5} \oplus F^{10}_{1,3}\\
\mathcal{Y}_{1,4,5}&\!=\!F^1_{4,5} \oplus F^9_{1,5} \oplus F^{10}_{1,4},\ \  \mathcal{Y}_{2,3,4}\!=\!F^2_{3,4} \oplus F^7_{2,4} \oplus F^9_{2,3}\\
\mathcal{Y}_{2,3,5}&\!=\!F^2_{3,5} \oplus F^7_{2,5} \oplus F^{10}_{2,3},\ \  \mathcal{Y}_{2,4,5}\!=\!F^2_{4,5} \oplus F^9_{2,5} \oplus F^{10}_{2,4}\\
\mathcal{Y}_{3,4,5}&\!=\!F^7_{4,5} \oplus F^9_{3,5} \oplus F^{10}_{3,4}.
\end{align*}
After this round user 1, 2, 7, 9, and 10 can successfully decode their required files. 
Then, in the second round of delivery, the BS serves users 3 and 4 from cache 1 and user 8 from cache 2.  \textcolor{black}{For this round, we have $\mathbf{V}^2 \neq \hat{\mathbf{V}}$, $\mathcal{A}_2= [3,4]$, and the demand vector $\mathbf{d}^{\mathbf{V}^2}=[3,4,8]$. Thus, the BS transmits the following seven subpackets based on set $\mathcal{Q}_{t+1}^{\mathcal{C}}$,}
\begin{align*}
\mathcal{Y}_{1,3,4}&=F^3_{3,4} \oplus F^8_{1,4},\ \ \mathcal{Y}_{1,3,5}=F^3_{3,5} \oplus F^8_{1,5}\\
\mathcal{Y}_{2,3,4}&=F^4_{3,4} \oplus F^8_{2,4},\ \ \mathcal{Y}_{2,3,5}=F^4_{3,5} \oplus F^8_{2,5}\\
\mathcal{Y}_{1,4,5}&=F^3_{4,5},\ \  \mathcal{Y}_{2,4,5}=F^4_{4,5}, \ \ \mathcal{Y}_{3,4,5}=F^8_{4,5}.
\end{align*}
After this round users 3, 4, and 8 can successfully decode their required files.
Next, in the final round of delivery, the BS serves users 5 and 6 from cache 1. For this round, \textcolor{black}{we have $\mathbf{V}^3 \neq \hat{\mathbf{V}}$, $\mathcal{A}_3= [2,3,4]$, and the demand vector $\mathbf{d}^{\mathbf{V}^3}=[5,6]$.
We can see that subpacket $\mathcal{Y}_{3,4,5}$ will not serve any user, thus, the BS transmits the following six subpackets based on set $\mathcal{Q}_{t+1}^{\mathcal{C}}\backslash (3,4,5)$,}
\begin{align*}
\mathcal{Y}_{1,3,4}&=F^5_{3,4},\ \ \mathcal{Y}_{1,3,5}=F^5_{3,5},\ \ \mathcal{Y}_{1,4,5}=F^5_{4,5} \\
\mathcal{Y}_{2,3,4}&=F^6_{3,4},\ \ \mathcal{Y}_{2,3,5}=F^6_{3,5},\ \ \mathcal{Y}_{2,4,5}=F^6_{4,5}.
\end{align*}
After this round users 5 and 6 can successfully decode their required files. This completes the content delivery phase, which results in a delivery time of $T(\mathbf{V})= \frac{20}{9}$.

\section{Numerical Evaluation}\label{Sec:4}
We know from Theorem~\ref{th:tavg} that it is computationally expensive to numerically evaluate the exact average delay even for small system parameters. \textcolor{black}{This is due} to the fact that \textcolor{black}{such evaluation} would require \textcolor{black}{the generation of the set} $\mathcal{V}$ of all possible cache population vectors $\mathbf{V}$, which \textcolor{black}{corresponds to} the so-called weak composition problem, and \textcolor{black}{where the} cardinality of $\mathcal{V}$ is known to \textcolor{black}{grow} exponentially with system parameters $K$ and $\Lambda$. \textcolor{black}{Instead}, we proceed to numerically evaluate our results by using the \emph{sampling-based numerical} (SBN) approximation method, where we generate a large set $\mathcal{V}_1$ of randomly generated cache population vectors $\mathbf{V}$ based on cache population intensities $\mathbf{p}$, and approximate $\overline{T}(t)$ as
\begin{align}\label{eq:tavgs} 
\overline{T}(t)\approx \frac{1}{|\mathcal{V}_1|}\sum_{\mathbf{V}\in \mathcal{V}_1} T(\mathbf{V}),
 \end{align}
where $T(\mathbf{V})$ is given in~\eqref{eq:tv}. Then, the corresponding approximate performance is evaluated by comparing it with the achievable average delay for the uniform cache size from \cite{malik_tcom21}, which is given as
\begin{align}\label{eq:tuni}
\overline{T}_{uni}\approx \frac{1}{|\mathcal{V}_1|}\sum_{\mathbf{V}\in \mathcal{V}_1}  \sum_{\lambda=1}^{\Lambda-t} \mathbf{L}(\lambda) \frac{{\Lambda-\lambda \choose t} }{{\Lambda \choose t}}, 
 \end{align}
\BS{where $\mathbf{L} = sort(\mathbf{V})$  is the sorted (in descending order) version of the cache load vector $\mathbf{V}$.}

Figure~\ref{fig:Tavg_sbn} compares the SBN approximation from~\eqref{eq:tavgs} with the SBN approximation from~\eqref{eq:tuni} for $|\mathcal{V}_1| = 10000$, $K=400$, $\Lambda=10$, and $\mathbf{p}=[0.2, 0.2 , 0.15, 0.1,0.1, 0.05, 0.05, 0.05,$ $  0.05, 0.05]$, where $\mathcal{V}_1$ is generated based on cache population intensities vector $\mathbf{p}$. \textcolor{black}{This figure} highlights the significant gain that can be achieved by allocating the cache capacity according to the cache population intensities as our scheme significantly outperforms the uniform cache size based coded caching scheme. For the same parameter setup, 
Figure~\ref{fig:gamma} compares the cache capacity allocations of our scheme with uniform cache capacity allocations for various capacity budgets $t$. \textcolor{black}{This new figure illustrates how our scheme allocates} cache capacity in proportion with cache population intensities.
\begin{figure}[t]
\centering
 \includegraphics[width=0.99\linewidth]{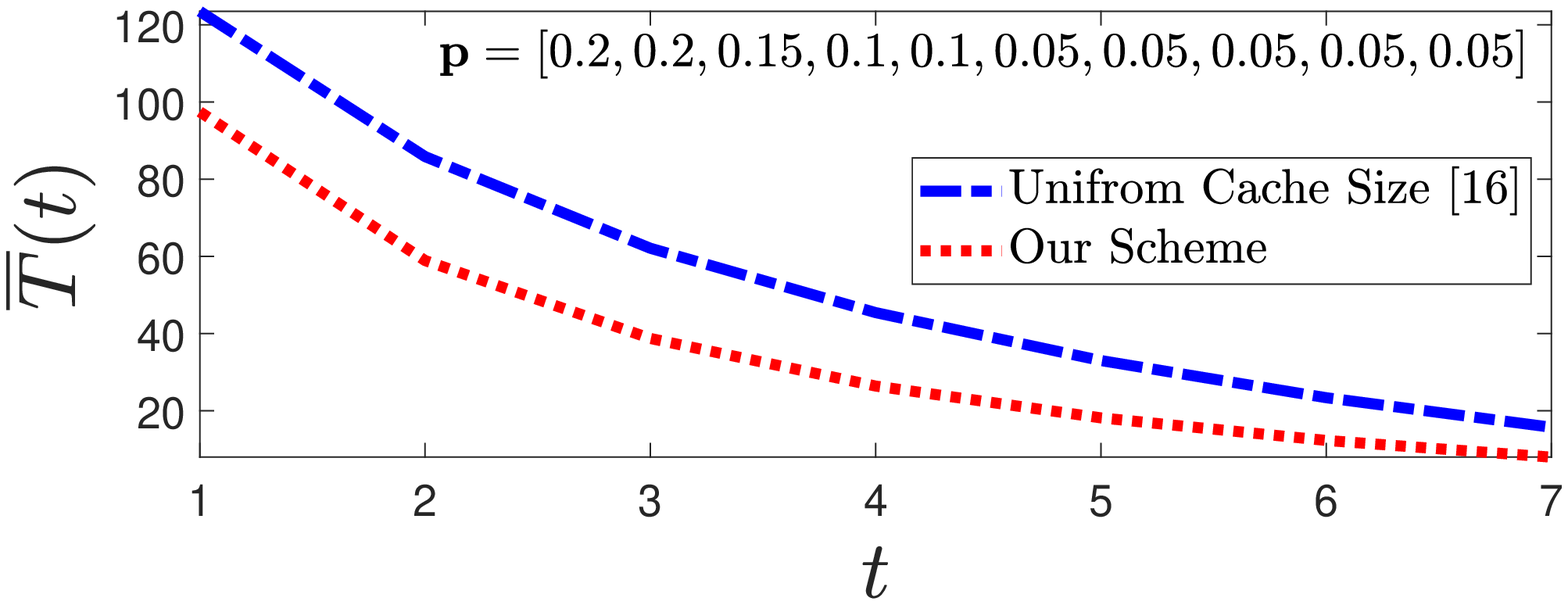}
\caption{$\overline{T}(t)$ from \eqref{eq:tavgs} vs $\overline{T}_{uni}$ from \eqref{eq:tuni}.}
\label{fig:Tavg_sbn}
\end{figure}

\section{Conclusion}
The work explored the coded caching problem in stochastic shared-cache networks, where each user can appear within the coverage area of one of $\Lambda$ caches with a given probability, and \textcolor{black}{in this context proposed} a scheme that optimizes the storage allocation of caches under a cumulative cache-size constraint. The novel scheme alleviates the adverse effect of cache-load-imbalance by significantly \textcolor{black}{ameliorating the detrimental} performance deterioration due to randomness. \textcolor{black}{Furthermore, for each and every instance of the coded caching problem, our scheme substantially alleviates -- compared to the best known state-of-art --- the well-known subpacketization bottleneck.}

\begin{figure}[t]
\centering
 \includegraphics[width=0.99\linewidth]{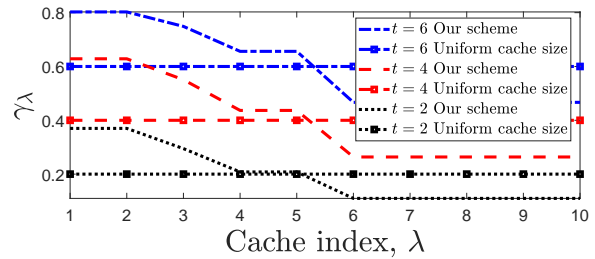}
\caption{Cache capacity allocation.}
\label{fig:gamma}
\end{figure}
\bibliographystyle{IEEEtran}

\bibliography{IEEEabrv,mainJ}





\ifCLASSOPTIONcaptionsoff
  \newpage
\fi



%
%



\end{document}